\documentclass[pra,aps,superscriptaddress,twocolumn,nofootinbib,longbibliography,allowtoday,a4paper]{quantumarticle}

\pdfoutput=1
\usepackage[T1]{fontenc}
\usepackage[numbers,comma,sort&compress]{natbib}
\usepackage{graphicx,
color,
amsmath,
amsfonts,
enumerate,
amsthm,
amssymb,
mathtools,
enumitem,
thmtools,
mathrsfs,
hyperref,
subfigure,
mathdots,
enumitem,
centernot,
bm,
soul,
bbm}
\usepackage[normalem]{ulem}
\usepackage{booktabs}
\usepackage{url}
\usepackage{float}
\usepackage{tkz-euclide}
\usepackage{csquotes}
\usepackage{bbold}
\usepackage{pgfplots}
\usepackage[capitalise, noabbrev]{cleveref}
\usepackage{dsfont}
\usepackage{cancel}
\usepackage{enumitem}
\usepackage{multirow}
\hypersetup{colorlinks=true,linkcolor=blue,citecolor=blue,urlcolor=blue}
\usepackage[margin=2cm]{geometry}
\usepackage[utf8]{inputenc}

\usetikzlibrary{decorations.pathmorphing}

\usepackage{listings}
\usepackage{color}

\usepackage{tikz,pgfplots}
\usetikzlibrary{quantikz}
\usetikzlibrary{arrows,shapes,positioning}
\usetikzlibrary{decorations.markings}

\theoremstyle{plain}

\newcommand{\loc}{L}

\newtheorem{theorem}{Theorem}

\newtheorem{prop}[theorem]{Proposition}
\newtheorem{obs}[theorem]{Observation}
\newtheorem{cor}[theorem]{Corollary}

\theoremstyle{remark}

\theoremstyle{definition}

\definecolor{codegreen}{rgb}{0,0.6,0}
\definecolor{codegray}{rgb}{0.5,0.5,0.5}
\definecolor{codepurple}{rgb}{0.58,0,0.82}
\definecolor{backcolour}{rgb}{0.95,0.95,0.92}

\lstdefinestyle{mystyle}{
    backgroundcolor=\color{backcolour},   
    commentstyle=\color{codegreen},
    keywordstyle=\color{magenta},
    numberstyle=\tiny\color{codegray},
    stringstyle=\color{codepurple},
    basicstyle=\footnotesize,
    breakatwhitespace=false,         
    breaklines=true,                 
    captionpos=b,                    
    keepspaces=true,                 
    numbers=left,                    
    numbersep=5pt,                  
    showspaces=false,                
    showstringspaces=false,
    showtabs=false,                  
    tabsize=2
}
 
\usepackage{tikz}
\usetikzlibrary{arrows}

\DeclareMathOperator{\diag}{diag}

\usepackage{array}
\usepackage{rotating}
\usepackage{tabularx}
\usepackage{geometry}

{\protect}

\newcommand{\id}{\mathbf{1}}

\renewcommand{\emph}[1]{{\it #1}}
\newcommand{\tr}{\mathrm{Tr}}
\newcommand{\ketbra}[2]{|{#1}\rangle \! \langle {#2}|}

\newcommand{\R}{\mathds{R}} 
\def\E{ {\cal E} }

\newcommand{\dyad}[2]{| #1\rangle \langle #2|}

\usepackage {tikz}
\usetikzlibrary {positioning}

\usepackage{setspace}
\usepackage{orcidlink}

\usepackage{xcolor}

\newcommand*{\colorboxed}{}
\def\colorboxed#1#{%
  \colorboxedAux{#1}%
}
\newcommand*{\colorboxedAux}[3]{
  \begingroup
    \colorlet{cb@saved}{.}
    \color#1{#2}
    \boxed{
      \color{cb@saved}
      #3
    }
  \endgroup
}

\makeatletter
\renewcommand*\env@matrix[1][\arraystretch]{
  \edef\arraystretch{#1}
  \hskip -\arraycolsep
  \let\@ifnextchar\new@ifnextchar
  \array{*\c@MaxMatrixCols c}}
\makeatother

\begin{document}

\title{Certifying nonlocal properties of noisy quantum operations}

\author{Albert Rico${}^{\orcidlink{0000-0001-8211-499X}}$}
\affiliation{
Faculty of Physics, Astronomy and Applied Computer Science, Institute of Theoretical Physics, Jagiellonian University,
30-348 Krak\'{o}w, 
Poland}
\author{Mois\'es Bermejo Mor\'an${}^{\orcidlink{0000-0003-1441-0468}}$}
\affiliation{
Faculty of Physics, Astronomy and Applied Computer Science, Institute of Theoretical Physics, Jagiellonian University,
30-348 Krak\'{o}w, 
Poland}
\affiliation{
Laboratoire d’Information Quantique, Université libre de Bruxelles, Belgium
}
\author{Fereshte Shahbeigi${}^{\orcidlink{0000-0001-9991-6112}}$}
\affiliation{
Faculty of Physics, Astronomy and Applied Computer Science, Institute of Theoretical Physics, Jagiellonian University,
30-348 Krak\'{o}w, 
Poland}
\affiliation{RCQI, Institute of Physics, Slovak Academy of Sciences, D\'{u}bravsk\'{a} cesta 9, 84511 Bratislava, Slovakia}
\author{Karol \.{Z}yczkowski${}^{\orcidlink{0000-0002-0653-3639}}$}
\affiliation{
Faculty of Physics, Astronomy and Applied Computer Science, Institute of Theoretical Physics, Jagiellonian University,
30-348 Krak\'{o}w, 
Poland}
\affiliation{
Center for Theoretical Physics, Polish Academy of Science,
02-668 Warszawa, Poland}

\date{\today}
\maketitle

\begin{abstract}
Certifying quantum properties from the probability distributions they induce is an important task for several purposes. While this framework has been largely explored and used for quantum states, its extrapolation to the level of channels started recently in a variety of approaches. In particular, little is known about to what extent noise can spoil certification methods for channels.

In this work we provide a unified methodology to certify nonlocal properties of quantum channels from the correlations obtained in prepare-and-measurement protocols: our approach gathers fully and semi-device-independent existing methods for this purpose, and extends them to new certification criteria. In addition, the effect of different models of dephasing noise is analysed. Some noise models are shown to generate nonlocality and entanglement in special cases. In the extreme case of complete dephasing, the measurement protocols discussed yield particularly simple tests to certify nonlocality, which can be obtained from known criteria by fixing the dephasing basis. These are based on the relations between bipartite quantum channels and their classical analogues: bipartite stochastic matrices defining conditional distributions.

\end{abstract}


\section{Introduction}\label{sec:Int}

Quantum theory predicts phenomena that cannot be explained with local hidden-variable models~\cite{bell1964einstein, bell1966problem}. This property, known as nonlocality, is a resource for certain protocols in quantum information and computation~\cite{wiesner1983conjugate,pironio2010random}. These are implemented by using quantum states describing physical systems, which are processed and manipulated through quantum operations. Information is obtained by performing measurements, which provide  a probability distribution in terms of relative frequencies of outcomes. 

Therefore, it is of primal importance to identify nonlocal properties of both quantum states~\cite{werner1989quantum,barrett2002nonsequential} and the quantum operations acting on them~\cite{ChSteerNLbeyond_Hoban2018Games,rosset2018resource,CorrCompMeas_Selby2023Games}, from the nonlocal behaviors in the outcome probability distributions~\cite{RobustnessLNoise_Massar2002,QNLin2-3levRobustnessWN_Acin2002,BritoQuantBellNLtrDist_2018}. Bell games are linear tests for measurement protocols that identify nonlocality in quantum states, by analysing the correlations obtained through classical inputs and outputs~\cite{bell1964einstein,bell1966problem,BrunnerReviewBellNonloc_2014}. By using quantum inputs, these can be modified to detect nonlocality of all entangled states~\cite{BuscemiNonlocHV_2012}. Although these methods have been employed to study nonlocal properties of quantum states for several decades~\cite{bell1964einstein,bell1966problem,BuscemiNonlocHV_2012,BrunnerReviewBellNonloc_2014}, it was more recently that similar schemes were proposed to identify nonlocality in quantum channels~\cite{ChSteerNLbeyond_Hoban2018Games,rosset2018resource,ResNCchAssem_Zjawin2023Games,CorrCompMeas_Selby2023Games,schmid2020type,Liu_ResTherQChan2020}.

However, quantum operations in realistic setups suffer from the phenomenon of {\em decoherence}: due to imperfect isolation, 
internal degrees of freedom of the system are coupled with the environment
~\cite{peres1997quantum,Nielsen_Chuang_2010}. 
This typically leads to a full or partial loss of quantum properties~\cite{schlosshauer2007quantum}, and therefore identifying nonlocality in decoherent setups becomes a challenge~\cite{CharNonClassCorr_Rahman2022}. In particular, quantum systems are affected by {\em dephasing} noise~\cite{Breuer-opensysytems,wiseman2009qmeascontrol}, which acts on a quantum state by damping the off-diagonal terms in the density matrix while leaving intact the diagonal terms in a fixed basis~\cite{devetak2005capacitydeco,DArrigo_2007_CapDepChann,Bradler2010TroffCapHadChann}.

In this paper we provide a framework to certify nonlocal properties of quantum channels, in the presence of dephasing noise. 
To this end, we use 
measurement protocols depicted in Fig.~\ref{fig:DecoGames} to formulate (semi-)device-independent criteria to certify nonlocality of quantum channels. 
This approach includes and extends existing methods~\cite{ChSteerNLbeyond_Hoban2018Games,schmid2020type} and allows us to identify the role of dephasing noise in the nonlocal properties of quantum channels: while local dephasing cannot generate nonlocality, there are cases in which a global dephasing operation can generate both entanglement and nonlocality. 
Our analysis in the case of complete decoherence yields simple criteria to detect nonlocality in quantum channels through a geometric insight into the transition from quantum to classical operations.



\section{Preliminaries}
\subsection{Quantum and classical states and channels}

The set of quantum states $\rho$ acting on a Hilbert space of dimension $d$ is represented by $\mathcal{M}_d$. The set of classical states, namely the $d$-dimensional probability simplex, is denoted by $\Delta_d$. A classical channel transforming probability distributions is a \emph{stochastic matrix}, that is, a matrix with non-negative entries that sum up to one in each column. The set of stochastic matrices with the dimension of inputs and outputs equal to $d$ is represented by $\mathcal{S}_d$. Analogously, completely positive (CP) and trace-preserving (TP) linear maps, $\E_{A_0\rightarrow A_1}$, also known as quantum operations or \emph{quantum channels} \cite[§10.3]{bengtsson2012geometry}, transform quantum states. Their inputs and outputs are density matrices acting on a Hilbert space $\mathcal{H}_{A_0}$ and $\mathcal{H}_{A_1}$ with dimensions $d_{A_0}$ and $d_{A_1}$, respectively. When the input and output systems are held by the same party, we will denote by $A$ the joint system $A_0A_1$ and use the notation $\E_{A}$, or simply $\E$, if no confusion arises.

We denote by $J^{\mathcal{E}}_{A_0A_1}$ (or simply $J^\E$) the $d_{A_0} d_{A_1}\times d_{A_0} d_{A_1}$ \emph{Choi matrix} acting on $\mathcal{H}_{A_0}\otimes\mathcal{H}_{A_1}$, which defines a quantum channel $\E_A$ through the Choi-Jamio\l kowski isomorphism~\cite{choi1975completely,jamiolkowski1972linear},
\begin{equation}
\label{eq:CJ}
    \E_{A}(\rho_{A_0})=\tr_{A_0}(J_{A_0A_1}(\rho_{A_0}^T\otimes\id_{A_1}))\,.
\end{equation}
Here, $\id_{X}$ denotes the identity operator on $\mathcal H_{X}$ and the superscript $T$ the transpose operation.
The complete positivity and trace preserving (CPTP) conditions require that $J_{A_0A_1}\geq 0$ and
\begin{equation}\label{eq:TP}
    J_{A_0} := \tr_{A_1}(J_{A_0A_1})=\id_{A_0},
\end{equation} where we denote the partial trace by omitting the traced subsystems. 

Following~\cite{Supermaps_Chiribella2008,Superchannels_Gour2019}, we denote a \emph{quantum superchannel} as a bipartite channel $\Xi_{AB}$ with the special property that it transforms Choi matrices to Choi matrices. That is, a quantum channel $\E_{A_1\rightarrow B_0}$ is transformed according to
\begin{equation}
    \Xi_{A_0B_0\rightarrow A_1B_1}(\E_{A_1\rightarrow B_0})=\E'_{A_0\rightarrow B_1}.
\end{equation}
For the above to hold, the Choi matrix of $\Xi$ should satisfy~\cite{Superchannels_Gour2019}
\begin{subequations}\label{eq:superchannel}
\begin{align}
    J^\Xi&\geq 0\label{eq:SuperchPos}\\
    J^\Xi_{A_0B_0}&=\id_{A_0B_0}\label{eq:SuperchTP}\\
    J^\Xi_{AB_0}&=J^\Xi_{A_0A_1}\otimes\id_{B_0}/d_{B_0}\,.\label{eq:SuperchTPP}
\end{align}
\end{subequations}

Similarly, a classical superchannel $\Gamma$ \cite{ClassSupermaps_Hasenohrl2022}, transforming stochastic matrices of size $d_{B_0}\times d_{A_1}$ to those of size  $d_{B_1}\times d_{A_0}$, corresponds with a bipartite stochastic matrix of dimension $d_{A_1}d_{B_1}\times d_{A_0}d_{B_0}$. To be a valid superchannel, it has to additionally satisfy
\begin{equation}
\label{eq:sup-nonsignaling}
    \sum_{b_1}\Gamma_{a_1b_1,a_0b_0}=\sum_{b_1}\Gamma_{a_1b_1,a_0b'_0}\,.
\end{equation}
Physically, equations ~\eqref{eq:SuperchTPP} and \eqref{eq:sup-nonsignaling} impose \emph{nonsignaling conditions}:
the outputs cannot signal to the inputs; see also Eqs.~\eqref{eq:qNSab} and \eqref{eq:qNSba}.
The set of quantum and classical superchannels acting on $d$-dimensional channels are respectively denoted by $\Theta_d$ and $\mathcal{T}_d$.

\subsection{Nonlocality of bipartite distributions}
\label{sec:non-loc}

A bipartite probability distribution describes \emph{independent} events when it is a product of local distributions, $p(a,b|x,y) = p(a|x) p(b|y)$. The convex closure of these forms the \emph{Bell local} distributions (\loc{})~\cite{bell1966problem,BrunnerReviewBellNonloc_2014}, decomposing as
\begin{equation}
\label{eq:ProbSep}
    p(a,b|x,y)=\sum_\lambda p_\lambda \, p(a|x,\lambda) p(b|y,\lambda)\,,
\end{equation}
where $p_\lambda\geq 0$ and $\sum_\lambda p_\lambda=1$.
Otherwise, the distribution is called {\em Bell nonlocal}.

In quantum theory, probabilities $p(a,b|x,y)$ are obtained through Born's rule,
\begin{equation}\label{eq:ProbBorn}
 p(a,b|x,y) = \tr(M^{a|x}\otimes N^{b|y}\rho_{AB})
\end{equation}
for some shared state $\rho_{AB}$ and positive operator-valued measurements (POVM) $M^{x}$ and $N^{y}$. That is, $M^{a|x}, N^{b|y}\geq 0$ for each effect and $\sum_a M^{a|x} = \id_A$, $\sum_{b} N^{b|y} = \id_B$ \cite[§10.1]{bengtsson2012geometry}.
Distributions obtained from Eq.~\eqref{eq:ProbBorn} are called {\em quantum} (Q). 

A probability distribution is called {\em nonsignaling} (NS) if it has well-defined marginals, i.e., it satisfies
\begin{subequations}\label{eq:ProbNS}
\begin{align}
    p(a|x) =\sum_b p(a,b|x,y)& =  \sum_b p(a,b|x,y') \, ,\\
    p(b|y) =\sum_a p(a,b|x,y)& =  \sum_a p(a,b|x',y) \, . 
\end{align}
\end{subequations}
The sets \loc{}, Q and NS are convex and form a chain of inclusions
\begin{equation}
\operatorname{\loc{}}\subset\operatorname{Q}\subset\operatorname{NS}\, .\\
\end{equation}
Thus, they can be separated with linear functionals
\begin{equation}\label{eq:BellFunct}
\gamma(p):=\sum_{xyab}\gamma_{ab,xy}p(ab|xy)
\end{equation}
on the space of bipartite probability distributions, called \emph{Bell functionals}. 
One denotes $\gamma_{S}:=\max_{p\in S}\gamma(p)$ the maximum value of a Bell functional $\gamma$ over the set of distributions inside a selected subset $S$ (typically \loc{}, Q or NS). For example, the Clauser-Horne-Shimony-Holt (CHSH) \cite{CHSHineq1969} functional $\gamma_{ab,xy}=(-1)^{a+b}(-1)^{xy}$ satisfies $\gamma_{\operatorname{\loc{}}} = 2$, $\gamma_{\operatorname{Q}} = 2\sqrt{2}$ and $\gamma_{\operatorname{NS}} = 4$ \cite{tsirel1987quantum, BrunnerReviewBellNonloc_2014}.

\subsection{Nonlocality in quantum states}
A \emph{Bell protocol} is a device-independent local measurement protocol where a bipartite probability distribution is obtained from a bipartite quantum state through Eq.~\eqref{eq:ProbBorn}. Assuming quantum theory, the only subsets of bipartite distributions that can be obtained in this way are local (\loc{}) or quantum (Q). Separable quantum states can only produce local correlations, and therefore observing nonlocal quantum correlations is an indicator of the presence of entanglement~\cite{HyllusWitFromBell_2005,OtfriedED}. A quantum state $\rho_{AB}$ which cannot produce nonlocal correlations in the bipartition $A|B$ will be called a {\em local state} in such bipartition.

Remarkably, not all entangled states can produce correlations outside of \loc{} through 
Eq.~\eqref{eq:ProbBorn}, namely there exist local but entangled states~\cite{werner1989quantum,PPTnonloc_Vertesi2014}. However, a subtle difference exists between separable and entangled states in certain measurement protocols: it was shown by Buscemi that all entangled states $\rho_{AB}$ produce correlations that cannot be obtained from separable states in a semi-device-independent protocol with trusted quantum inputs~\cite{BuscemiNonlocHV_2012},
\begin{equation}\label{eq:BuscemiGame}
    p(ab|xy)=\tr(\tau^x_R\otimes\rho_{AB}\otimes\omega^y_S\cdot M^a_{AR}\otimes N^y_{BS})\,.
\end{equation}
Here the inputs are fixed sets of quantum states $\{\tau^x\}$ and $\{\omega^y\}$ used in ancillary systems $R$ and $S$. 
Distributions that cannot be obtained in this way from separable states are known as {\em Buscemi nonlocal}. For example, consider a state that is entangled but cannot produce distributions out of \loc{}  through Eq.~\eqref{eq:ProbBorn} ~\cite{werner1989quantum}. Through Eq.~\eqref{eq:BuscemiGame}, such a state produces distributions that are inside the local polytope \loc{}, but  are Buscemi nonlocal and therefore denote entanglement. 
Quantum measurement protocols obtained through Eq.~\eqref{eq:BuscemiGame} are called \emph{Buscemi protocols} to distinguish them from \emph{Bell protocols} in Eq.~\eqref{eq:ProbBorn}. 

\subsection{Nonlocality of bipartite channels}
\label{sec:non-loc-state-channel}
Bipartite classical channels are characterized by bipartite stochastic matrices~\cite{InterconvNLcorrClasChan_Jones2005,InfProcGPTClasChan_Barrett2007,OprFramewNLClasChan_Gallego2012,NLRTmeasClasChan_deVicente2014}, which are by definition isomorphic to bipartite correlations. That is to say,  by arranging a conditional probability distribution $p(ab|xy)$ in a matrix with columns labeled by $xy$, one obtains a bipartite stochastic matrix. 
This implies that the nonlocal structure of bipartite correlations defines accordingly a nonlocal structure in the set of bipartite stochastic maps. In particular, consider bipartite stochastic matrices of the form
\begin{equation}\label{eq:SepStoch}
S_{AB}=\sum_\lambda p_\lambda T_A^{\lambda}\otimes R_B^{\lambda}\,,
\end{equation}
where $T_A^\lambda\in\R^{d_{A_1}\times d_{A_0}}$ and $R_B^{\lambda}\in\R^{d_{B_1}\times d_{B_0}}$ are stochastic matrices. These stochastic matrices correspond with local probability distributions that decompose as Eq.~\eqref{eq:ProbSep}, and thus we call them local stochastic matrices. 

In the framework of quantum channel nonlocality, we distinguish the following three classes of channels.
\begin{itemize}
 \item {\em Local operations assisted by shared randomness} (LOSR): Operations $\E_{AB}$ which can be represented as the convex combination of {\em local operations} -- tensor product of quantum channels,  
\begin{equation}\label{eq:LOSR}
 \mathcal{E}_{AB}=\sum_\lambda p_\lambda \mathcal{E}_A^\lambda\otimes\mathcal{E}_B^\lambda\,,
\end{equation}
where $p_\lambda\geq 0$, $\sum_\lambda p_\lambda=1$ and $\mathcal{E}_{A(B)}^\lambda$ are quantum channels. These are the free operations in the resource theory of both Bell and Buscemi nonlocality~\cite{Geller_2014,deVicente_2014,schmid2020type,rosset-type-independent-prl} and are analogous to separable bipartite quantum states. Channels outside the set of LOSR are termed \emph{nonlocal} and can be compared to entangled states.

 \item {\em Localizable}, or {\em local operations and shared entanglement} (LOSE): Operations $\mathcal{E}_{AB}$ where Alice and Bob can do local operations assisted by a shared quantum state $\tau_{EF}$~\cite{BGNP01causal,schmid2021postquantum}. Namely,
 \begin{equation}\label{eq:Localizable}
     \mathcal{E}(\rho_{A_0B_0})\!=\!\tr_{E_1F_1}\! \big(\mathcal{E}_{AE}\otimes\mathcal{E}_{BF}(\rho_{A_0B_0}\otimes\tau_{E_0F_0})\!\big).
 \end{equation}
  
  \item  {\em Causal}, or {\em quantum nonsignaling (QNS)}: Operations $\mathcal{E}_{AB}$ where Bob's (Alice's) output cannot be influenced by Alice's (Bob's) local operations~\cite{BGNP01causal,Leung_NSCodesSDP_2015}. The marginals of their four-partite Choi matrices $J^\mathcal E_{A_0B_0A_1B_1}$ satisfy 
  \begin{subequations}
  \begin{align}
  J^{\mathcal{E}}_{A_0B_0B_1}&=\id_{A_0}\otimes J^{\mathcal{E}}_{B_0B_1}/d_{A_0},\label{eq:qNSab}\\ J^{\mathcal{E}}_{A_0A_1B_0}&= J^{\mathcal{E}}_{A_0A_1}\otimes\id_{B_0}/d_{B_0}\,.\label{eq:qNSba}
  \end{align}
  \end{subequations}
  Analogous to nonsignaling correlations in Eq.~\eqref{eq:ProbNS}, nonsignaling channels have also well-defined marginal channels proportional to $J^{\mathcal{E}}_{A_0A_1}$ and $J^{\mathcal{E}}_{B_0B_1}$.
\end{itemize}

We have the chain of inclusions (see Fig.~\ref{fig:SubsetsDeco})
\begin{equation}
\operatorname{LOSR} \subset \operatorname{LOSE} \subset \operatorname{QNS} \, .\label{subeq:chain}
\end{equation}

\subsection{The effect of dephasing noise in quantum states and channels}
Here we will consider quantum noise producing decoherence through a {\em dephasing channel}  $\mathcal{D}$, which leaves invariant the occupations of the input state in a preferred basis $\{\ket{i}\}$, $\bra{i}\mathcal{D}(\rho)\ket{i}=\bra{i}\rho\ket{i}$, while damping or inducing phases on the coherences. By imposing that a dephasing channel is a CPTP linear map, one verifies that it acts as~\cite{AlgebraDephasing_Kye1995,AlgebraDephasing2_Li1997,levick2017quantumPrivacySchurProdChannels,Puchala_Dephasing2021}
\begin{equation}\label{eq:DephChan}
    \mathcal{D}^G(\rho)=\rho\odot G
\end{equation}
for some Gram matrix $G$, i.e., a non-negative matrix with diagonal entries equal to one.
Here $\odot$ denotes the Schur product
(also called  entry-wise or Hadamard product). That is,  $\mathcal{D}^G(\rho)_{ij}=\rho_{ij}G_{ij}$ where
$|\mathcal{D}^G(\rho)_{ij}|=|\rho_{ij}G_{ij}|\leq|\rho_{ij}|$.

Dephasing noise can lead to the full or partial loss of the resource of coherence~\cite{ResCoh_Winter2016}, which plays a crucial role in most quantum protocols~\cite{Teleport_Bennett1993,EntSwapExp_Wei1998,BB84QKD_bennett1988}. The reason why the definition of the dephasing noise is basis dependent, 
is because quantum protocols take place by measuring the inputs and outputs in a fixed basis which is determined by physical principles, e.g. the energy eigenbasis in Hamiltonian dynamics or the basis where measurements are performed in a given protocol~\cite{Teleport_Bennett1993,EntSwapExp_Wei1998,BB84QKD_bennett1988, KamilCohQChan_2018}.

Similarly as for quantum states, the decoherent noise acting on a quantum channel $\mathcal{E}$ through dephasing leaves invariant the diagonal entries of the Choi matrix in a preferred basis and damps the cross-terms~\cite{Puchala_Dephasing2021}. A Choi matrix $J^\E_{A_0A_1}$ describing a channel $\E_A$ is converted into $J^\E_{A_0A_1}\odot G_{A_0A_1}$ where $G_{A_0A_1}$ is a Gram matrix with some additional structure~\cite{Puchala_Dephasing2021}. The most general form of such dephasing transformation is realized by pre- and post-processings along with memory \cite{Supermaps_Chiribella2008}. Assuming that there is no memory, dephasing noise acts as $\E':=\mathcal{D}^{G'}\circ\E\circ\mathcal{D}^{G}$ where $\circ$ means concatenation.  Then, the Choi matrix $J^\E_{A_0A_1}$ of a quantum channel $\E_A$ is mapped to~\cite{Puchala_Dephasing2021}
\begin{equation}\label{eq:GramNoMemoryDeph}
    J^{\mathcal E^\prime}_{A_0A_1}=J^\E_{A_0A_1}\odot (G_{A_0}\otimes G'_{A_1})\, .
\end{equation}

\subsection{Complete decoherence: from quantum to classical}
Consider now complete decoherence occurring by a completely dephasing channel $\mathcal{D}^\id$, acting as in Eq.~\eqref{eq:DephChan} where $G=\id$. Then, a quantum state $\rho\in\mathcal{M}_d$ loses all off-diagonal terms and becomes a probability vector $p=\diag(\rho)\in\Delta_d$. Similarly, when complete decoherence occurs in both the input and output systems of a quantum channel $\mathcal{E}_A$, the resulting operation is given by the diagonal of the Choi matrix \cite{Puchala_Dephasing2021}. That is, one has
\begin{equation}
    J^{\mathcal{D}^\id\circ\mathcal{\E}\circ\mathcal{D}^\id}= \iota [\diag(J^{\mathcal{E}})] \, .
    \label{eq:noise}
\end{equation}
where $\iota$ embeds a vector into the diagonal of a matrix. The CPTP conditions on $\mathcal{E}_{A}$ imply that the vector $\diag(J^{\mathcal{E}})$ can be reshaped into a $d_{A_1}\times d_{A_0}$ stochastic matrix $S^{\mathcal{E}}$, which is a linear map between the $d_{A_1}$- and $d_{A_0}$-dimensional probability simplices. The stochastic map $S^{\mathcal{E}}$, known as classical action~\cite{KamilCohQChan_2018,shahbeigi-log-convex,Puchala_Dephasing2021,QAdvSimStoch_Kamil2021,QEmbStoch_Fereshte2023}, will be called here the \emph{decoherent action} of $\mathcal{E}$~\cite{Puchala_Dephasing2021} to emphasize that it corresponds to a possibly nonlocal channel arising from the effects of decoherence. The decoherent action can be obtained from Kraus operators $K_i$ of $\mathcal{E}$ through $S^{\mathcal{E}}=\sum_i K_i\odot K^*_i$, where * denotes the complex conjugation~\cite{GeoQuantStates2006}. 

The effect of complete decoherence on quantum sets discussed above is summarized in Fig.~\ref{fig:SuperDeco}. These sets are obtained from the set $\mathcal{M}_{d^4}$ of four-partite states in different ways. Bipartite and monopartite states are obtained by partial trace. The probability simplices of corresponding sizes are obtained from these sets under decoherence. Bipartite and monopartite quantum channels are obtained from four-partite and bipartite states through the cross-section of Eq.~\eqref{eq:TP}. In particular, a generic quantum state $\sigma_{A_0B_0A_1B_1}$ with full-rank gives a generic quantum channel through
$J_{A_0B_0A_1B_1}=Y \sigma_{A_0B_0A_1B_1} Y$ where $X=\tr_{A_0A_1}\sigma_{A_0B_0A_1B_1}$
and $Y= \id \otimes 1/\sqrt{X}$ imposes the partial trace condition~\cite{kukulski2021generatingRandChann}. 
Then a classical map  (stochastic matrix) is obtained due to decoherence 
as in Eq.~\eqref{eq:noise}
from the diagonal entries of the Choi matrix $J$
describing the bipartite channel. The transition from channels to superchannels can be seen as an extended Choi-Jamio{\l}kowski isomorphism, as superchannels are isomorphic to one-way quantum nonsignaling channels through Eq.~\eqref{eq:SuperchTPP}. These decohere to one-way nonsignaling distributions through Eq.~\eqref{eq:sup-nonsignaling}. 

Thus one identifies classical subsets embedded in quantum subsets at different levels: general distributions, conditional distributions and nonsignaling distributions are embedded in quantum states, channels and superchannels. In this way, recent studies on the geometry of the 
 sets of local and nonlocal correlations~\cite{GeoQCorr_Goh2018,CharNonClassCorr_Rahman2022} are directly linked to the
 research on geometry of quantum states~\cite{GeoQuantStates2006,bengtsson2012geometry,Eltschka2021shapeofhigher}
 and geometry of quantum entanglement
\cite{kus2001geometry,pucha2012restricted}.
Furthermore, 
analysis of the relative volume of the set
of bipartite local correlations \cite{cabello2005how,wolfe2012quantum,duarte2018concentration}
can be compared with the problem
of estimating the relative volume of the set
of bipartite separable states \cite{karol1998volume,slater1999apriori,milz2014volumes,lovas2017invariance,sauer2021entanglement}.

\section{Channel nonlocality from local measurements}
\begin{figure*}[tbp]
    \centering
    \includegraphics[scale=0.85]{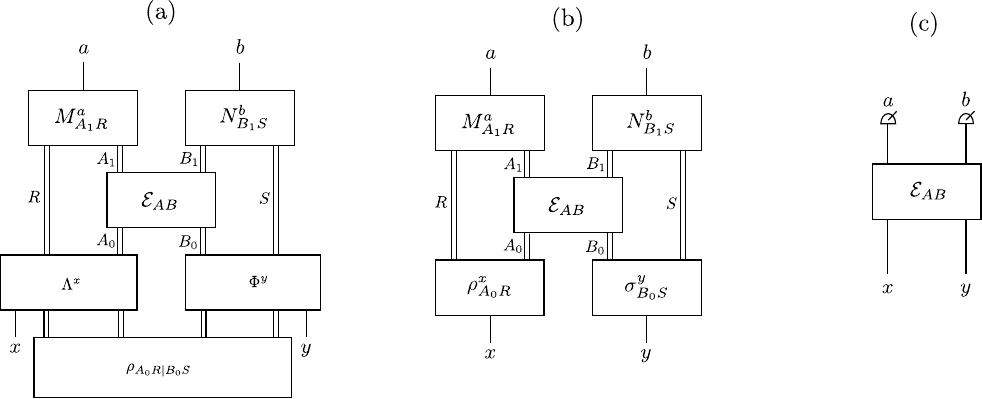}
    \caption{{\bf Schemes to obtain bipartite conditional probability distributions} (bipartite stochastic matrices) from bipartite channels: (a) A shared entangled state is processed with local operations (with ancillary systems) to produce the inputs which undergo the channel $\E$ to be analysed, and then local measurements are performed to obtained the probabilities of Eq.~\eqref{eq:CorrGameGen}. (b) Product states prepared locally undergo a bipartite channel and then local measurements are performed to obtain the distribution in Eq.~\eqref{eq:CorrGame}. This case can be seen as bipartite distributions being obtained after local superchannels are performed in the form of $\Xi_A\otimes\Xi_B[\E_{AB}]$. (c) A bipartite distribution in Eq.~\eqref{eq:CAchannel} is obtained with inputs in product states spanning a preferred basis, which undergo a bipartite channel and are then projective-measured in the same preferred basis.}
    \label{fig:DecoGames}
\end{figure*}

\subsection{Local measurement protocols with shared entanglement}

Bipartite distributions can be obtained from bipartite quantum channels by choosing input states and measuring their conditional outputs~\cite{ChSteerNLbeyond_Hoban2018Games,CorrCompMeas_Selby2023Games,ResNCchAssem_Zjawin2023Games,ResThNCcomcauAssem_Zjawin2023Games}, as shown in Fig.~\ref{fig:DecoGames}. In this spirit, consider the probability distribution obtained by the measurement protocol of Fig.~\ref{fig:DecoGames}~(a),
\begin{align}
    p(a,b|x,&y)=\tr\Big (\big ( M^a_{RA_1}\otimes N^b_{SB_1}\big )\cdot\nonumber\\
    &\E_{AB}\otimes\text{id}_{RS}[\Lambda_{A_0R}^x\otimes\Phi_{B_0S}^y(\rho_{ABRS})]\Big )\,.\label{eq:CorrGameGen}
\end{align}
Here, $\text{id}_{RS}$ denotes the identity channel in $RS$, $\mathcal E$, $\Lambda^x$ and $\Phi^y$ generic quantum channels, $M^a$ and $N^b$ POVMs and $\rho$ a quantum state. The subindices indicate their respective systems.
The following generalizes existing results in detecting non-LOSR quantum operations~\cite{ChSteerNLbeyond_Hoban2018Games,schmid2020type},  allowing for nonproduct input states --Fig.~\ref{fig:DecoGames} (a).
\begin{prop}\label{prop:LocconditionsGen}
Let $p(a, b|x, y)$ be a probability distribution obtained from a bipartite quantum channel $\E_{AB}$ through Eq.~\eqref{eq:CorrGameGen}. For any input channels $\Lambda^x$ and $\Phi^y$ and for any measurements $M^a$ and $N^b$, the following holds:
\begin{enumerate}[label=(\roman*)]
    \item If $\E_{AB}$ is in LOSR, then for any input state $\rho_{AR|BS}$ which is local in the bipartition $AR|BS$, $p(a,b|x,y)$ is in \loc{}.
    \item If $\E_{AB}$ is in LOSE, then for any shared state $\rho$, $p(a,b|x,y)$ is in Q.
    \item If $\E_{AB}$ is in QNS, then for any shared state $\rho$, $p(a,b|x,y)$ is in NS.
\end{enumerate}
\end{prop}
\begin{proof}
\textbf{Proof of (i).} For any channel $\E$, we denote its adjoint by $\E^\dagger$, which satisfies $\tr\left(\E[X]Y\right)=\tr\left(\E^\dagger[Y]X\right)$ for any $X$ and $Y$. It is enough to consider an extremal LOSR operation $\E_{AB}=\E_A\otimes\E_B$, since the rest of the proof follows by convex combinations. Define the operators
\begin{subequations}\label{eq:update-povmsLOSR}
\begin{align}
\tilde M^{a|x}:=&{\Lambda^x_{AR}}^\dagger\circ(\E^\dagger_{A}\otimes\text{id}_R)[M^a_{RA_1}]\, ,\\
\tilde N^{b|y}:=&{\Phi^x_{BS}}^\dagger\circ(\E^\dagger_{B}\otimes\text{id}_S)[N^b_{SB_1}]\, .
\end{align}
\end{subequations}
Since the adjoint $\E^\dagger$ of a CPTP map $\E$ is completely positive and unital, i.e., it satisfies $\E^\dagger(\id)=\id$~\cite{GeoQuantStates2006}, the operators $\Tilde{M}^{a|x}$ ($\Tilde{N}^{b|y}$)  are positive semidefinite and their summation over $a$ ($b$) is the identity for any $x$ ($y$). Thus, we can rewrite Eq.~\eqref{eq:CorrGameGen} as
\begin{equation}
    p(a,b|x,y)=\tr\Big (\rho_{ABRS}\cdot\tilde{M}^{a|x}_{RA_1}\otimes \tilde{N}^{b|y}_{SB_1}\Big )\label{eq:CorrGameGenLocProofDual}\,,
\end{equation}
which is a standard Bell game. Now it becomes evident by definition, that if the state $\rho_{ABRS}$ is local in the bipartition $AR|BS$, then a local distribution is obtained.

\textbf{Proof of (ii).}
Consider an LOSE channel, which factorizes as in Eq.~\eqref{eq:Localizable}. 
Similarly as in the proof of item (i), define the measurement operators
\begin{subequations}\label{eq:update-povms-prop2Gen}
\begin{align}
\tilde M^{a|x}:=&({\Lambda^x_{AR}}^\dagger\otimes\text{id}_{E})\circ\\
& (\E^\dagger_{AE}\otimes\text{id}_R)[M^a_{RA_1}\otimes\id_{E_1}]\, ,\nonumber\\
\tilde N^{b|y}:=&({\Phi^x_{BS}}^\dagger\otimes\text{id}_{F})\circ\\
& (\E^\dagger_{BF}\otimes\text{id}_S)[N^b_{SB_1}\otimes\id_{F_1}]\, .\nonumber
\end{align}
\end{subequations}
Using Eqs.~\eqref{eq:update-povms-prop2Gen}, one verifies that Eq.~\eqref{eq:CorrGameGen} reads
\begin{align}
    p(a,b|x,y) & = \tr \big ( \tilde M^{a|x}_{ARE} \otimes \tilde N^{b|y}_{BSF} \cdot \rho_{ARBS}\otimes\tau_{EF} \big )\,.\label{eq:p(abxy)Localizable}
\end{align}
This proves that by applying an LOSE channel one can only recover quantum correlations. 

\textbf{Proof of (iii).} By imposing in Eq.~\eqref{eq:CorrGameGen} that $\E$ satisfies $J^\E_{A_0B_0B_1}=\id_{A_0}\otimes J_{B_0B_1}/d_{A_0}$, we obtain
\begin{align}
    \sum_ap(a,b|x,y)=&\sum_a\tr\big ((\id_{B_0}\otimes N^b_{SB_1})\cdot\nonumber\\
    &(J_{B_0B_1}\otimes\id_S)(\rho^y_{B_0S}\otimes\id_{B_1})\big )\label{eq:ProofNSProp}
\end{align}
where $\rho^y_{B_0S}:=\Phi^y(\tr_{A_0R}(\rho_{ARBS}))$, which does not depend on $x$. Similarly, by imposing $J^\E_{A_0A_1B_0}= J_{A_0A_1}\otimes\id_{B_0}/d_{B_0}$ we obtain that $\sum_bp(a,b|x,y)$ does not depend on $y$ and therefore $p(a,b|x,y)$ is a nonsignaling distribution.
\end{proof}
Note that Proposition~\ref{prop:LocconditionsGen} provides a criterion to detect quantum channels outside of the sets LOSR, LOSE and QNS which is fully device-independent, in the sense that none of the quantum devices (the shared state $\rho$, the input channels $\Lambda^x$ and $\Phi^y$, and the measurements $M^a$ and $N^b$) is trusted. 
\begin{figure}[!ht]
    \centering
    \includegraphics[scale=0.025]{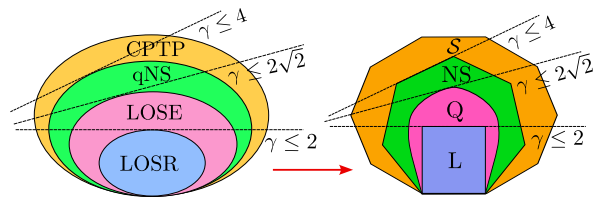}
    \caption{{\bf Subsets of distributions obtained from bipartite operations.} {\em Left.} Inclusions of the following subsets of quantum channels (CPTP): quantum nonsignaling (QNS), local operations assisted by the shared resources of entanglement (LOSE) and randomness (LOSR). {\em Right.} Inclusions of the subsets of the polytope of bipartite stochastic matrices ($\mathcal{S}$): nonsignaling (NS), quantum (Q) and local (\loc{}) correlations. 
    Bell functionals $\gamma$ bounding each subset of bipartite distributions bound the corresponding subsets of bipartite channels, through Prop.~\ref{prop:LocconditionsGen}, Cor.~\ref{prop:Locconditions} and Prop.~\ref{prop:CAiffCorrSubs}. Here we depict the linear cuts delimited by some Bell functional.
    }
    \label{fig:SubsetsDeco}
\end{figure}

\subsection{Local measurement protocols without entanglement}
Consider the protocol in Fig.~\ref{fig:DecoGames}~(b),
\begin{align}
p(a,b|x,y)=&\tr\Big (\big ( M^a_{RA_1}\otimes N^b_{SB_1}\big )\cdot\nonumber\\
&\E_{AB}\otimes\text{id}_{RS}(\rho_{A_0R}^x\otimes\sigma_{B_0S}^y)\Big )\,,\label{eq:CorrGame}
\end{align}
which was introduced in~\cite{ChSteerNLbeyond_Hoban2018Games}. This can be obtained from Eq.~\eqref{eq:CorrGameGen} by considering a product state $\rho_{ARBS}=\dyad{00}{00}_{AR}\otimes\dyad{00}{00}_{BS}$, so that the channels $\Lambda_{AR}^x$ and $\Phi_{BS}^y$ act as state preparation channels producing input states $\rho_{AR}^x$ and $\sigma_{BS}^y$. LOSR channels produce local probability distributions from Eq.~\eqref{eq:CorrGame}, as shown in~\cite{ChSteerNLbeyond_Hoban2018Games}. In items $(ii)$ and $(iii)$ of the following Corollary we extend this result to other classes of channels.
\begin{cor}\label{prop:Locconditions}
Let $p(a, b|x, y)$ be a probability distribution obtained from a bipartite quantum channel $\E_{AB}$ through Eq.~\eqref{eq:CorrGame}. The following holds, for any choice of input states and measurements:
\begin{enumerate}[label=(\roman*)]
    \item If $\E_{AB}$ is LOSR, then  $p(a,b|x,y)$ is in \loc{}.
    \item If $\E_{AB}$ is LOSE, then  $p(a,b|x,y)$ is in Q.
    \item If $\E_{AB}$ is QNS, then  $p(a,b|x,y)$ is in NS.
\end{enumerate}
\end{cor}

\begin{proof}
\textbf{Proof of (i).} Notice that probability distributions in Eq.~\eqref{eq:CorrGame} from LOSR channels give local probability distributions. This was noted earlier in~\cite[see Fig. 8]{ChSteerNLbeyond_Hoban2018Games}.

\textbf{Proof of (ii).}
Similarly as in Prop.~\ref{prop:LocconditionsGen}, define
\begin{subequations}\label{eq:update-povms}
\begin{align}
\tilde M^{a|x}=\tr_{RA_0}\big(&(\rho^{x}_{RA_0}\otimes\id_{E_0}).\\
& (\E^\dagger_{AE}\otimes\text{id}_R)[M^a_{RA_1}\otimes\id_{E_1}]\big)\, ,\nonumber\\
\tilde N^{b|y}=\tr_{SB_0}\big(&(\sigma^{y}_{SB_0}\otimes\id_{F_0}).\\
& (\E^\dagger_{BF}\otimes\text{id}_S)[N^b_{SB_1}\otimes\id_{F_1}]\big)\, .\nonumber
\end{align}
\end{subequations}
One verifies that $p(a,b|x,y)$ in Eq.~\eqref{eq:CorrGame} reads
\begin{equation}
    p(a,b|x,y)  = \tr \big ( \tilde M^{a|x}_E \otimes \tilde N^{b|y}_F \cdot \tau_{EF} \big )\,.\label{eq:p(abxy)Localizable}
\end{equation} 
{\bf Item (iii)} follows from Proposition~\ref{prop:LocconditionsGen}.
\end{proof}

Two remarks are now given. On the one hand, note that the converse of Corollary~\ref{prop:Locconditions} is not true. In particular, channels denoted as {\em local-limited} in~\cite{ChSteerNLbeyond_Hoban2018Games} always produce local correlations for any choice of input states and measurements. Examples are LOSE channels where Alice and Bob share an entangled state that cannot produce nonlocal correlations~\cite{werner1989quantum}. On the other hand, Eq.~\eqref{eq:p(abxy)Localizable} implies that a bipartite LOSE channel produces nonlocal correlations if the shared state $\tau_{EF}$ in Eq.~\eqref{eq:Localizable} is nonlocal in some Bell protocol. This in turn dictates that if the shared state in one realization of an LOSE channel is Bell nonlocal, the channel cannot be realized by some shared local state.

\subsection{A Buscemi-like characterization of LOSR operations}

The results of Buscemi~\cite{BuscemiNonlocHV_2012} extend to more general processes~\cite{schmid2020type}. For a fixed channel $\E$ and Bell functional $\gamma$, we denote as $\gamma_{\max}(\E)$ the maximal value of $\gamma(p)$ in Eq.~\eqref{eq:BellFunct} that can be achieved when the conditional probability distribution $p$ is obtained through Eq.~\eqref{eq:CorrGame} by a fixed set of input states $\rho^x$ and $\sigma^y$ over all possible measurements $M^a$ and $N^b$. That is,
\begin{equation}
    \gamma_{\max}(\E):=\max_{\{M^a,N^b\}}\sum_{xy,ab}\gamma_{ab,xy} p(a,b|x,y)\,.
    \label{eq:gammax}
\end{equation}

The setup depicted in Fig.~\ref{fig:DecoGames}~(b) can be seen as a particular case of so-called {\em semi-quantum games}, which are known to characterize convertibility of general processes through free operations in a generalized theory of nonlocality~\cite{schmid2020type}. In particular, the setup described in Eq.~\eqref{eq:CorrGame} and Fig.~\ref{fig:DecoGames}~(b) provides semi-quantum games characterizing local transformations between LOSR quantum channels. The following is a particular case of~\cite[Corollary 1]{schmid2020type}:

\begin{prop}\label{prop:BuscemiChannels}
    Given two bipartite quantum operations $\E_{AB}$ and $\E'_{A'B'}$, $\E_{AB}$ can be transformed into $\E'_{A'B'}$ by a convex combination of product superchannels if and only if for all Bell functionals $\gamma$ and for all sets of input states, $\gamma_{\max}(\E)\geq\gamma_{\max}(\E')$, where $\gamma_{\max}$ is given by Eq.~\eqref{eq:gammax}.
\end{prop}
\begin{proof}
Starting from the more general results in~\cite{schmid2020type}, the proof of the {\bf only if} direction can be easily seen as follows. If $\E$ can be converted into $\E'$ by a product of superchannels (or a convex combination of them), then we can absorb such product of superchannels in the state preparations and measurements of semiquantum games, described in Eq.~\eqref{eq:CorrGame} and depicted in Fig.~\ref{fig:DecoGames}~(b). Then it becomes clear that $\gamma_{\max}(\E)\geq\gamma_{\max}(\E')$ for all functionals $\gamma$. The proof for the 
{\bf if} direction is based on the observation that Eq.~\eqref{eq:CorrGame} gives a distribution after local superchannels are performed in the form of $\Xi_A\otimes\Xi_B[\E_{AB}]$, thus preserving the set of LOSR operations~\cite{schmid2020type}. For self-consistency of the paper, an explicit constructive proof is given in Appendix~\ref{app:ProofBuscemiChannels}.
\end{proof}
Similarly as in~\cite{BuscemiNonlocHV_2012}, this means that a quantum channel $\E$ is not LOSR if and only if for some Bell functional $\gamma$ and quantum inputs, the value $\gamma_{\max}(\E)$ in Eq.~\eqref{eq:gammax} cannot be attained with LOSR operations. 
The derivations in Appendix~\ref{app:ProofBuscemiChannels} provide explicit necessary and sufficient conditions to certify non-LOSR operations. Consider a set of fixed states $\{\rho^x_{RA_0}\}$ and $\{\sigma^y_{SB_0}\}$ that are linearly independent and span the whole input spaces of linear operators acting on $\mathcal{H}_{RA_0}$ and $\mathcal{H}_{SB_0}$, and assume $\mathcal{H}_{A_0}\approx\mathcal{H}_R$ and $\mathcal{H}_{B_0}\approx\mathcal{H}_S$. Let $\{\Psi^a_{RA_1}\}$ and $\{\Psi^b_{SB_1}\}$ be projective measurements on a generalized Bell basis and denote
\begin{align}
p_\E(ab|xy):=\tr(\Psi^a_{RA_1}&\otimes \Psi^b_{SB_1}\cdot \label{eq:FixBuscemiChannels}\\
\ (\E_{AB}&\otimes\text{id}_{RS})[\rho^x_{RA_0}\otimes\sigma^y_{SB_0}]).\nonumber
\end{align}
This setup allows to formulate the problem of whether or not a channel $\E_{AB}$ is LOSR as follows: A quantum channel $\E_{AB}$ is not an LOSR operation if and only if for any possible LOSR operation $\E'_{A'B'}$ and for any choice of measurements $M^a_{RA_1}$ and $N^b_{SB_1}$, the distribution $p_\E(ab|xy)$ of Eq.~\eqref{eq:FixBuscemiChannels} cannot be obtained via $\tr(M^a_{RA_1}\otimes N^b_{SB_1}(\E'_{AB}\otimes\operatorname{id}_{RS})[\rho^x_{RA_0}\otimes\sigma^y_{SB_0}])$. By absorbing the adjoint ${\E'_{AB}}^\dag$ in the measurements $M^a\otimes N^b$, namely defining $\tilde{M}^a\otimes \tilde{N}^b:=\E^\dag\otimes\text{id}(M^a\otimes N^b)$, we derive the following criterion, analogous to Buscemi's characterization of entangled states ~\cite{BuscemiNonlocHV_2012}.
\begin{cor}\label{cor:BuscemiChannels}
A bipartite quantum channel $\E_{AB}$ is not LOSR if and only if for some linear functional $\gamma$,
\begin{equation}
    \gamma(p_\E)>\max_{\tilde{M}^a,\tilde{N}^b}\sum_{ab,xy}\gamma_{ab,xy}\tr(\tilde{M}^a\rho^x)\tr(\tilde{N}^b\sigma^y)\,,
    \label{eq:gamma_pE}
\end{equation}
where we denote $\gamma(p_\E):=\sum_{ab,xy}\gamma_{ab,xy}p_\E(ab|xy)$.
\end{cor}
This provides a separation between the set of LOSR operations and the rest, which relies only in the input/outcome probabilities and the trusted set of input states $\rho^x$ and $\sigma^y$. Therefore, it constitutes a measurement-device-independent characterization of the set of LOSR operations. Note that the left hand side of Eq.~\eqref{eq:gamma_pE} is given by the set of input states and the channel under consideration, and the right hand side can be upper-bounded with semidefinite programming~\cite{berta2021semidefinite}.

\begin{figure*}[tbp]
    \centering
    \includegraphics{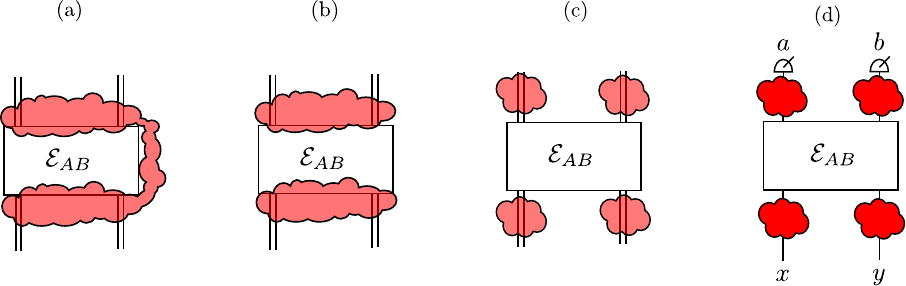}
    \caption{{\bf Models of decoherence occurring by dephasing noise under consideration.} (a) Dephasing occurs both in the inputs and the outputs as a superchannel $\Xi_{AB}[\E_{AB}]$. Its action on the Choi matrix is given by a Gram matrix $G_{A_0A_1B_0B_1}$ without any tensor product assumption through a bipartite Eq.~\eqref{eq:DephChan} with a bipartite structure. (b) Nonlocal dephasing occurs as independent pre- and post-processing by composition with a dephasing channel, $\E_{AB}^\mathcal{D}=\mathcal{D}^G\circ\E\circ\mathcal{D}^{G'}$. Its action on the Choi matrix is given by a product of Gram matrix $G_{A_0B_0}\otimes G_{A_1B_1}$ through Eq.~\eqref{eq:GramNoMemoryDeph} with a bipartite structure. (c) Local dephasing occurs at each party independently, according to Eq.~\eqref{eq:noise_local}. (d) Local complete dephasing occurs, where the input and output states are diagonal density matrices. The last scheme representing complete decoherence corresponds to Fig.~\ref{fig:DecoGames} (c), since local projective measurements in a fixed basis occur both before and after the channel is applied.}
    \label{fig:DecoModels}
\end{figure*}
\section{The role of dephasing noise}

\subsection{Gradual dephasing}
Quantum channels, and in particular the protocols described above to certify nonlocality, can suffer from decoherence. Here we ask how dephasing noise affects the nonlocal properties of a channel. Let us first consider local dephasing (Fig.~\ref{fig:DecoModels}~(c)), which is given by a local four partite Gram matrix of the form
\begin{equation}
G = G_{A_0} \otimes G_{B_0} \otimes G_{A_1} \otimes G_{B_1} \,.
\label{eq:noise_local}
\end{equation}

\begin{obs}\label{obs:LocDeph}
If a bipartite channel undergoing local dephasing noise, $\mathcal{D}_{A_1}\otimes\mathcal{D}_{B_1} \circ \E_{AB} \circ \mathcal{D}_{A_0}\otimes\mathcal{D}_{B_0}$, can produce correlations outside of \loc{} (Q, NS) from Eq.~\eqref{eq:CorrGame}, then the channel $\E_{AB}$ is not LOSR (LOSE, QNS).
\end{obs}
Indeed, note that the Schur product with a Gram matrix preserves both positive semidefiniteness, the trace, and identity resolution of any operators. Therefore, local dephasing is absorbed in the preparation and measurement of any protocol. This means that if we detect nonlocality on a channel which suffered from local dephasing, then we know that its noiseless version previous to dephasing is also nonlocal.

Let us now consider the case in which the two parties $A$ and $B$ sharing a quantum state $\rho_{AB}$ or implementing a quantum channel $\E_{AB}$ belong to a joint subsystem. In this case, dephasing can occur as a joint operation (Fig.~\ref{fig:DecoModels}~(a \& b)), and thus one might expect that nonlocal phenomena can occur.  This leads us to the following result.
\begin{prop}\label{prop:DecoEntanglement}
Dephasing noise acting jointly on a bipartite system has the following properties:
\begin{enumerate}
 \item A dephasing channel acting on a separable bipartite state can generate an entangled state.
 \item A dephasing superchannel acting on an LOSR operation can generate a nonlocal channel.
\end{enumerate}
\end{prop}
At the level of quantum states, an extreme example can be constructed as follows: 
consider the states $\ket{\pm{}}=(\ket{0}\pm{}\ket{1})/\sqrt{2}$. Note that the projector $\rho_{AB}=\dyad{++}{++}$ has homogeneous entries $\rho_{ij}=1/4$, thus it evolves under dephasing noise $\mathcal{D}_{AB}^G$ in Eq.~\eqref{eq:noise} given by some Gram matrix $G$ as
\begin{equation}\label{eq:DephOn+}
    \mathcal{D}^G_{AB}(\rho_{A_0B_0})=\dyad{++}{++}\odot G=G/4\,.
\end{equation}
Therefore, it is enough to show that there exists a quantum state that is both maximally entangled and a Gram matrix. Indeed, the Bell state $\ket{\psi}=(\ket{0+}+\ket{1-})/\sqrt{2}$ has density matrix
\begin{equation}\label{eq:BellGram}
    \dyad{\psi}{\psi}=\frac{1}{4}
    \begin{pmatrix}
        \phantom{-}1 & \phantom{-}1 & \phantom{-}1 & -1 \\
        \phantom{-}1 & \phantom{-}1 & \phantom{-}1 & -1 \\
        \phantom{-}1 & \phantom{-}1 & \phantom{-}1 & -1 \\
        -1 & -1 & -1 & \phantom{-}1 \\
    \end{pmatrix}
\end{equation}
which defines a Gram matrix $G=4\dyad{\psi}{\psi}$ with entries $G_{ij}=\bra{v_i}v_j\rangle$ where $\ket{v_1}=\ket{v_2}=\ket{v_3}=-\ket{v_4}=(1,1,1,-1)^T/2$. Since $\ket{\psi}$ is maximally entangled, it maximally violates CHSH inequality,
\begin{equation}
     \langle \sigma_Z\otimes(\sigma_++\sigma_-)+\sigma_X\otimes(\sigma_+-\sigma_-) \rangle_\psi=2\sqrt{2}\, ,
\end{equation}
where $\sigma_X$ and $\sigma_Z$ are Pauli matrices and $\sigma_{\pm}=(\sigma_X\pm\sigma_Z)/\sqrt{2}$.
From a resource-theoretic perspective, this is possible because the state $\ket{++}:=\ket{+}\otimes\ket{+}$ maximizes the resource of coherence (with respect to the computational basis)~\cite{QuantifCoh_Baumgratz2014,ResCoh_Winter2016}, from which a dephasing channel can generate entanglement. 

At the level of quantum channels, this example can be seen as follows: consider the identity channel $\text{id}$, with the bipartite input state $\ket{++}$. Although the identity channel is an LOSR operation and thus produces probability distributions in local measurement protocols like the one of Eq.~\eqref{eq:CorrGame} (Fig.~\ref{fig:DecoGames}), dephasing induced by coupling degrees of freedom with the environment can convert the channel $\text{id}$ into a channel capable of generating nonlocal quantum probability distributions. This is an example where dephasing occurs as a superchannel with independent pre- and post-processing~\cite{Puchala_Dephasing2021,Supermaps_Chiribella2008} (Fig.~\ref{fig:DecoModels}~(b)). 

An example where dephasing is induced by a more general superchannel (Fig.~\ref{fig:DecoModels}~(a)) is as follows: consider the Choi matrix $J^{U_L}$ of the local unitary $U_{L}=\id\otimes\sigma_X$. Let the dephasing superchannel be described by a $16\times16$ Gram matrix $G=G^\dag\geq 0$ with nonzero coefficients $G_{kk}=1$, $G_{12,15}=-G_{2,5}=1$, and $G_{2,12}=G_{5,12}=-G_{2,15}=G_{5,15}=i$ (and their Hermitian adjoints), where $i$ is the imaginary unit. This dephasing leads to the Choi matrix $J^{U}=J^{U_L}\odot G$ which defines the channel implementing the following unitary,
\begin{equation}
    U=\sum_{j=0}^1\dyad{j}{j}\otimes\sigma_X(i\sigma_Z)^j\,.
\end{equation}
Observe that this unitary can be converted into the C-NOT gate by local unitaries as follows,
\begin{equation}
    J_{NOT}=(\id\otimes H\sigma_X) U (\id\otimes H)^\dag\, ,
\end{equation}
and therefore $U$ can produce extremal signaling probability distributions in the protocol in Fig.~\ref{fig:DecoGames}~(b) when decoherence depicted in Fig.~\ref{fig:DecoModels}~(a) occurs.

Let us highlight that the examples above constitute a special case of dephasing channels and superchannels occurring to highly coherent states and channels, whose coherence only suffer from phase factors. Although these examples are consistent with the mathematical model of dephasing operations established in~\cite{Puchala_Dephasing2021}, generic dephasing transforms pure states into mixed states, and unitary channels into irreversible channels~\cite{Nielsen_Chuang_2010}. Yet, these examples are illustrative for the following comparison: observation~\ref{obs:LocDeph} captures the intuitive fact that decohrence might make a nonlocal quantum channel undetectable through measurement protocols, thus giving a {\em false negative} result to a test of nonlocality. On the contrary, Propopsition~\ref{prop:DecoEntanglement} shows that decoherence can also give rise to {\em false positive} results if the parties $A$ and $B$ perform a measurement protocol undergoing nonlocal noise.

To better understand the role of gradual decoherence aside from extremal cases, we consider two parametric types of dephasing models. On the one hand, $\mathcal{D}_q$ damps all off-diagonal terms equally, $\mathcal{D}_q(\rho)=(1-q)\rho+q\diag(\rho)$. On the other hand, $\mathcal{D}'_p$ keeps all entries except for the coherences to the fourth population level, namely $\mathcal{D}'_p(\rho)_{ij}=(1-2p)\rho_{ij}$ if either $i=4\neq j$ or $j=4\neq i$.

Starting from the maximally coherent separable state $\rho=\dyad{++}{++}$, the composition of these two noises reads
\begin{equation}\label{eq:DecoEntParameter}
    \mathcal{D}_q\circ\mathcal{D}'_p(\rho)=\frac{1}{4}
        \begin{pmatrix}
        1 & 1-q & 1-q & x \\
        1-q &1 & 1-q & x \\
        1-q & 1-q & 1 &  x \\
        x & x & x & 1 \\
    \end{pmatrix}
\end{equation}
with $x=(1-q)(1-2p)$. The negativity of this state for the parameters $p$ and $q$ is shown in Fig.~\ref{fig:ContourNegativity}.
\begin{figure}
    \centering
    \includegraphics{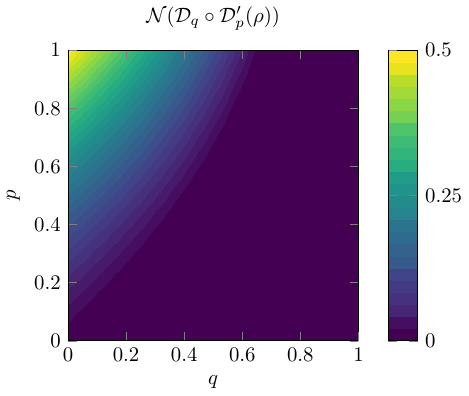}
    \caption{{\bf Entanglement generated by dephasing.} The parameters $p$ and $q$ quantify the strength of the dephasing channels $\mathcal{D}_q$ and $\mathcal{D}'_p$ acting on the maximally coherent state $\rho = \dyad{++}{++}$~\cite{QuantifCoh_Baumgratz2014,ResCoh_Winter2016}, which can be mapped to less coherent but entangled states. The color bar evaluates the negativity of the partial transpose of the resulting state, indicating maximal entanglement when taking the value $0.5$.}
    \label{fig:ContourNegativity}
\end{figure}

\subsection{Complete decoherence}
\begin{figure*}[!ht]
    \includegraphics[scale=0.27]{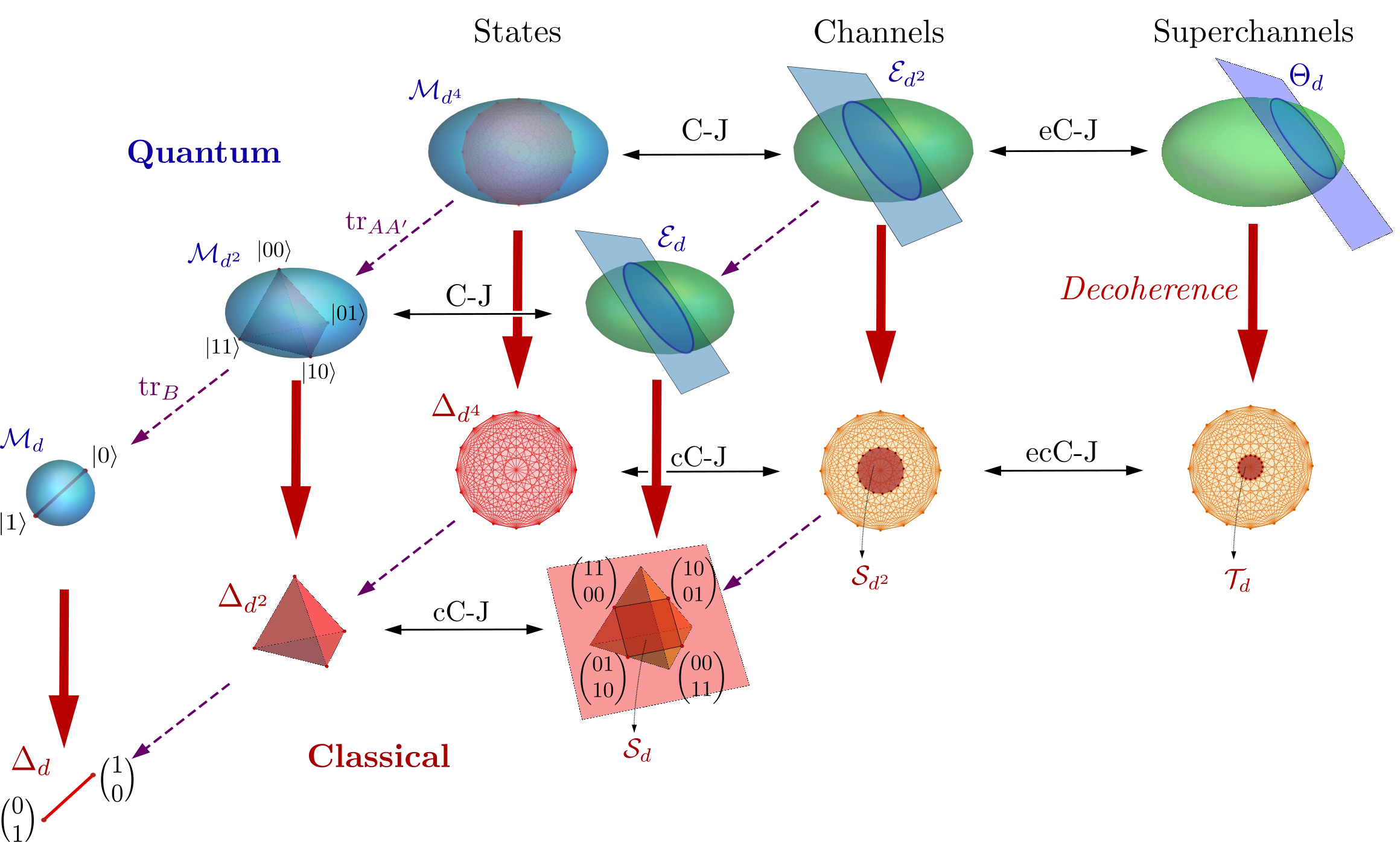}
    \caption{{\bf Complete decoherence of states, channels and superchannels.}  
    Blue convex sets in the left represent sets of quantum states consisting of $4$, $2$ and $1$ subsystems of size $d$ related by partial trace (directed dashed arrow). Bidirectional solid arrows denote the Choi-Jamio{\l}kowski isomorphism (C-J): 
    a cross section of the set $\mathcal{M}_{d^2}$
    defines  through Eq.~\eqref{eq:TP} the set
     ${\mathcal E}_d$ of
     quantum channels (blue section of the green convex set). An extension of the Choi-Jamio{\l}kowski isomorphism (eC-J) links the set of two-qudit channels ${\mathcal E}_{d^2}$  with 
     one-qudit superchannels 
     ${\Theta}_d$ 
     through Eq.~\eqref{eq:SuperchTPP}.
     Under decoherence (red thick arrows) these objects become classical polytopes depicted in red and contained inside larger orange polytopes:  
     a quantum state $\rho\in\mathcal{M}_d$ is mapped to a vector $p=\diag(\rho)\in\Delta_d$,
      a quantum channel $\E\in\mathcal{E}_{d}$ 
     is mapped into a stochastic matrix $S\in\mathcal{S}_d$,
       while 
      a quantum superchannel from $\Theta_d$ 
      is transformed  into an element
      of the set  $\mathcal{T}_d$
      of   classical supermaps. 
     The sets $\mathcal{M}_{d^2}$ of bipartite quantum states and $\mathcal{S}_{d^2}$ of bipartite
     stochastic matrices are both obtained from 
      the set
      $\mathcal{M}_{d^4}$ of four-partite states
     and inherit some of its properties.}
    \label{fig:SuperDeco}
\end{figure*}

Here we consider the case of complete decoherence depicted in Fig.~\ref{fig:SuperDeco}, where Choi matrices become diagonal. Recall that the set of stochastic matrices and the set of conditional probability distributions are isomorphic.  
For bipartite channels $\mathcal{E}_{AB}$, the decoherent action is represented by a bipartite stochastic matrix with entries
\begin{equation}\label{eq:CAchannel}
 S_{ab,xy}^{\mathcal{E}}=\bra{ab}\mathcal{E}(\dyad{xy}{xy})\ket{ab}\,,
\end{equation}
where $(a,b,x,y)\in [d_{A_1}] \times [d_{B_1}] \times [d_{A_0}] \times [d_{B_0}]$ define the canonical basis (Fig.~\ref{fig:SuperDeco}). Therefore, the decoherent action~\eqref{eq:CAchannel} defines probabilities $p(a,b|x,y)=S_{ab,xy}$ of obtaining classical outputs $ab$ conditioned to classical inputs $xy$, through the scheme of Fig.~\ref{fig:DecoGames} (c). In ~\cite[Propositions~3-5]{ChSteerNLbeyond_Hoban2018Games} it was shown that the sets of \loc{}, Q and NS correlations correspond to the sets of correlations that can be obtained respectively from LOSR, LOSE, and QNS channels as correlations in the inputs and outputs when all measurement bases are allowed. Since we are interested in dephasing noise, here we reformulate this result for a fixed basis
considering the decoherent action of Eq.~\eqref{eq:CAchannel}, as follows. 

\begin{prop}\label{prop:CAiffCorrSubs}
The following equivalences hold:
\begin{enumerate}[label=(\roman*)]
    \item The set of \emph{local (\loc{})} correlations coincides with the set of  decoherent actions of \emph{LOSR} channels.
    \item The set of \emph{quantum (Q)} correlations coincides with the set of  decoherent actions of \emph{LOSE} channels.
    \item The set of \emph{nonsignaling (NS)} correlations coincides with the set of decoherent actions of \emph{quantum nonsignaling} channels QNS. 
\end{enumerate}
\end{prop} 
Proposition~\ref{prop:CAiffCorrSubs} can be obtained from~\cite[Prop. 3, 4 \& 5]{ChSteerNLbeyond_Hoban2018Games}, by absorbing the basis choice into the existence conditions. A qualitative analysis of the channel nonlocality that is witnessed by the nonlocality of such correlations is carried out in a different work \cite[Prop. 1]{Rico2024ourshort}. For self-consistency, an explicit constructive proof is given in Appendix~\ref{app:CAiffCorrSubs}. 
This provides a simple criterion to detect certain classes of quantum channels, from local measurements in the computational basis: one certifies that a quantum channel $\E_{AB}$ is outside of the set of LOSR (LOSE, QNS) channels if its decoherent action $S_{AB}^{\mathcal{E}}$ is outside of the set of \loc{} (Q, NS) correlations. 
This criterion can be implemented in decoherent setups, as it involves only measurements in the computational basis. 

Let us remark that the converse is not true: there exist non-LOSR quantum channels with local decoherent action. As an example, consider the 2-qubit channel with unitary (single) Kraus operator $U = \sum_{i,j=0}^1\ket{\phi_{ij}}\bra{i+j,j}$, which maps the computational basis to the Bell states $\ket{\phi_{ij}}=\sigma_X^i\sigma_Z^j\otimes\id(\ket{00}+\ket{11})/\sqrt{2}$ where $\sigma_X$ and $\sigma_Z$ are Pauli matrices. Even though $U$ can generate entanglement, its decoherent action $S^{U}=U\odot U^*=\frac{1}{2}\left(\id^{\otimes 2} + \sigma_X^{\otimes 2}\right)$ is local. However, note that this criterion can be enhanced with any operations that preserve the set of LOSR. In the example above, the freedom of local measurement basis considered in~\cite{ChSteerNLbeyond_Hoban2018Games} is enough: the bipartite unitary $U$ can be transformed by local unitaries to $\tilde{U}=(H\otimes H)U(\id\otimes H)$, where $H$ is the Hadamard gate with $H\ket{0}=\ket{+}$ and $H\ket{1}=\ket{-}$. We have $S^{\tilde{U}}=\tilde{U}\odot \tilde{U}^*=\dyad{0}{0}+\dyad{1}{3}+\dyad{2}{2}+\dyad{3}{1}$ which defines extremal nonsignaling correlations.

Note that the nonsignaling conditions establish equality constraints, and thus define a measure zero subset of channels. 
Moreover, the sets of local and nonsignaling correlations form a polytope. Therefore, detecting channels with nonlocal decoherent action can be done with a linear program, asking whether there exists a convex combination of extremal local stochastic matrices decomposing the decoherent action. For local correlations, finitely-many Bell games~\cite{HViffCHSH2222_Fine1982,Class50years_Rosset2014,BrunnerReviewBellNonloc_2014} equivalently define the faces of their polytope. Bounding the set of quantum correlations is a more involved problem and relies on a convergent hierarchy~\cite{navascues2008convergent}. This problem has been tackled through analytical solutions for the smallest setting \cite{tsirel1987quantum, landau1988empirical, masanes2005extremal}, conjectured generalizations \cite{GeoQCorr_Goh2018,mikosnuszkiewicz2023extremal} and additional results using semidefinite programming techniques \cite{thinh2019geometric, thinh2023quantum}. Outer approximations are given by the maximum quantum value of Bell games~\cite{cirel1980quantum}.

By Proposition~\ref{prop:CAiffCorrSubs}, Bell functionals on the decoherent action of a quantum channel can be used to construct linear nonlocality witnesses, similar to entanglement witnesses~\cite{OtfriedED}. This mechanism is illustrated in Fig.~\ref{fig:SubsetsDeco}.
Let $\mathcal L$ be one of the three subsets of bipartite quantum channels: LOSR, LOSE, and QNS. Let $S$ be the corresponding set of bipartite probability distributions under the correspondence in Proposition~\ref{prop:CAiffCorrSubs}, namely, \loc{}, Q, and NS. Consider
\begin{equation}\label{eq:WitnessBell}
W_\mathcal L=\gamma_S/(d_{A_0}d_{B_0})\id-\Omega  \, ,
\end{equation}
where $\Omega$ is defined from the Bell functional through
\begin{equation}\label{eq:OmegaCHSH}
    \Omega=\sum_{x,y,a,b=0}^1 \gamma_{ab,xy}\dyad{xayb}{xayb}_{A_0A_1B_0B_1} 
\end{equation}
and $\gamma_{S}$ denotes its maximum value on the subset $S$. Thus, a channel $\E_{AB}$ with Choi matrix $J^\E$ and decoherent action $S^\E$ is outside $\mathcal L$ if
\begin{equation}\label{eq:CnlMeasure}
    \tr(J^\E W_\mathcal L) = \gamma_S - \sum \gamma_{ab,xy}S^\E_{ab,xy} < 0\,.
\end{equation}
For the CHSH functional, $\gamma_{ab,xy}=(-1)^{a+b}(-1)^{xy}$, it holds that $\gamma_{\operatorname{\loc{}}} = 2$, $\gamma_{\operatorname{Q}} = 2\sqrt{2}$ and $\gamma_{\operatorname{NS}} = 4$ \cite{tsirel1987quantum, BrunnerReviewBellNonloc_2014}.

By the isomorphism of Proposition~\ref{prop:CAiffCorrSubs}, any bipartite stochastic matrix in a distinguished subset of correlations  can be obtained from decoherence of a quantum channel in a corresponding subset. An explicit example is the convex hull of bipartite stochastic matrices (namely, bipartite probability distributions)  $pR +(1-p)S$ with
\begin{equation}
R  = \frac{1}{2}\begin{pmatrix}
1 & 1 & 1 & 0 \\
0 & 0 & 0 & 1 \\
0 & 0 & 0 & 1 \\
1 & 1 & 1 & 0 
\end{pmatrix} ,\,  S  = \frac{1}{2}\begin{pmatrix}
0 & 0 & 0 & 1 \\
1 & 1 & 1 & 0 \\
1 & 1 & 1 & 0 \\
0 & 0 & 0 & 1 \\
\end{pmatrix}
\label{eq:cross_section}
\end{equation}
which attains the maximal quantum value $\gamma_Q=2\sqrt{2}$ of CHSH for $p=(1+1/\sqrt{2})/2$. 
In Fig.~\ref{fig:cross_section} we show a two-dimensional slice of the polytope of distributions consisting of the convex hull of $R$, $S$ and $\id$ containing different subsets.  The sets of local (\loc{}) and nonsignaling (NS) distributions are bounded by linear programming, while the set of quantum (Q) distributions is approximated by the first two levels of the Navascu\'{e}s-Pironio-Ac\'{i}n (NPA) hierarchy~\cite{navascues2008convergent}. The subsets L, Q and NS of distributions are obtained by decoherence of the subsets LOSR, LOSE and QNS of channels.
\begin{figure}[tbp]
\centering
\includegraphics[width=0.35\textwidth]{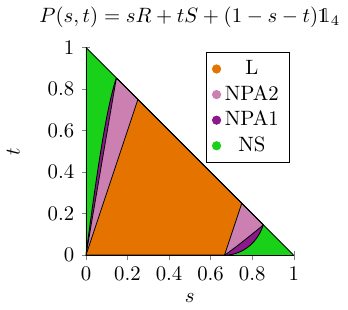}
\caption{{\bf Partial cross section of the space of bipartite distributions} spanned by $\id_4$, $R$ and $S$ in Eq.~\eqref{eq:cross_section}, $P(s, t) = sR + tS + (1-s-t)\id_4$ where $s+t \leq 1$. The distribution $\id_4$ is local, while $R$ and $S$ are Popescu-Rohrlich (PR) correlations~\cite{popescu1994quantum}, which are extremal nonsignaling. The orange region corresponds with distributions that are local, while dark and light purple regions denote first (NPA1) and second (NPA2) level of Navascu\'{e}s-Pironio-Ac\'{i}n (NPA) approximation \cite{navascues2008convergent} to the set $Q$ of quantum distributions (not plotted). Green region depicts the remaining nonsignaling distributions. Each of these subsets of distributions can be obtained by LOSR, LOSE and QNS channels under complete decoherence.}\label{fig:cross_section}
\end{figure}

\section{Conclusions and open questions}

We infer nonlocal properties of bipartite quantum channels based on the bipartite stochastic matrices produced by them under measurement protocols. 
In particular, our approach allows us to bound the sets of local operations assisted by shared randomness (LOSR), local operations assisted by shared entanglement (LOSE) and quantum nonsignaling channels (QNS), from the conditional  probabilities of local measurement outcomes.
This framework is suitable to experimental verification of nonlocality of bipartite channels from local measurements, even in the presence of complete decoherence. 
In contrast, we show that both entanglement and nonlocality of bipartite quantum channels can be enhanced under nonlocal gradual decoherence.

    \begin{figure} [h]
    \vspace{5mm}
    \centering
        \includegraphics{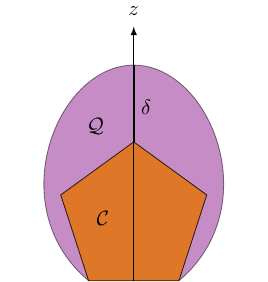}
        \caption{
{\bf First glimpse at the structure of multidimensional  convex sets  of {\em quantum} $(\mathcal Q)$ and {\em classical} $(\mathcal C)$}:
 states,  operations and  correlations. The transition from quantum to classical is induced by decoherence.
    The same figure can also represent sets of quantum algorithms or quantum technologies.
    The distance $\delta$ between the extremal points of both sets along 
the vertical axis $z$, denoting a selected utility function,
 represents the {\sl quantum advantage.}
    }
      \label{fig:states}
    \end{figure}

This work  relates the geometry of the set of bipartite quantum states, channels, superchannels, and their classical counterparts: probability vectors, stochastic matrices and classical superchannels.  In particular, several measures for nonlocality in the space of correlations have been proposed in the literature \cite{pironio2003violations, van2005statistical, barrett2005nonlocal, cabello2005how}. Through our framework, these naturally give insight into the geometry of quantum operations~\cite{karol1998volume}. At the same time, our work leaves several research lines open:
\begin{enumerate}

        \item It is unknown whether local-limited and quantum-limited channels as defined in~\cite{ChSteerNLbeyond_Hoban2018Games} can be detected by making use of a shared entangled state through Proposition~\ref{prop:LocconditionsGen}.

        \item   Corollary~\ref{cor:BuscemiChannels} provides a first step to understanding the set of functionals witnessing non-LOSR operations, analogous to entanglement witnesses for nonseparable states. However, further research is needed to better understand the performance of such witnesses in practice.
\end{enumerate}

The analysis presented for the case of complete decoherence highlights the role 
of Choi--Jamio{\l}kowski isomorphism 
to analyse geometric features of 
quantum correlations and 
channel nonlocality.
The standard form of the isomorphism
relates bipartite states with quantum operations, while its generalized version links 
four-partite quantum states with bipartite quantum channels. Due to decoherence the latter suffer a transition to bipartite stochastic matrices ${\cal S}_{d^2}$, which describe correlations between both systems. Therefore, the overall geometric
structure of the sets of classical and quantum states, channels and correlations is rooted in the geometric structure of the set ${\cal M}_{d^4}$ four-partite states (Fig.~\ref{fig:SuperDeco}).

Any generalized 
quantum theory can be recovered from basic principles that force
the corresponding set of states  
to be convex \cite{mielnik1974generalized}.
Hence, if one looks into the set of density matrices
from a large perspective,
one realizes that it 
forms a convex set 
\cite{Mielnik:1980hq}. 
A closer look 
reveals its convex subset of separable states, 
defined as classical mixtures of product states 
-- see Fig. \ref{fig:states}.
A related structure is characteristic 
also to the space of quantum maps and quantum correlations.
A similar schematic
picture can also be used to
describe the sets of quantum algorithms 
and quantum technologies and to visualize the notion of a 
{\em quantum advantage}:
the difference between the utility function describing, for instance,  
efficiency of a given protocol, 
optimized over all quantum and classical scenarios. 

\bigskip

{\em Acknowledgements:\,\,} We are thankful to Remigiusz Augusiak, Francesco Buscemi, Dariusz Chru\'{s}ci\'{n}ski, Jakub Czartowski, Felix Huber, Kamil Korzekwa, Miguel Navascu\'{e}s, Roberto Salazar, Anna Sanpera and Mario Ziman for discussions and for pointing to references. Support by the Foundation for Polish Science through TEAM-NET project  POIR.04.04.00-00-17C1/18-00
and by NCN QuantERA
Project No. 2021/03/Y/ST2/00193 and ERC Advanced Grant TAtypic, project number 101142236 is gratefully acknowledged. FS acknowledges projects DESCOM VEGA 2/0183/2 and DeQHOST APVV-22-0570.

\begin{widetext}
    
\appendix
\setcounter{secnumdepth}{1}

\section{Proof of Proposition~\ref{prop:BuscemiChannels}}\label{app:ProofBuscemiChannels}
Here we will prove the {\bf if} part of Proposition~\ref{prop:BuscemiChannels}. Namely, we will show that if $\E$ is more nonlocal than $\E'$ for all the games specified in Fig.~\ref{fig:DecoGames} (b), then it can be transformed to $\E'$ through some superchannels $\Xi=\sum\nu_i\Xi_i^A\otimes\Xi_i^B$.
\begin{proof}
We will follow a similar reasoning to the analogous result of Buscemi for bipartite quantum states~\cite{BuscemiNonlocHV_2012}. To this end, consider the extended Hilbert spaces 
\begin{align}
\mathcal{H}_{E_0}&=\mathcal{H}_{A_0}\otimes\mathcal{H}_{R_0}=\mathcal{H}_{A'_0}\otimes\mathcal{H}_{R'_0}\quad\text{and}\\
\mathcal{H}_{F_0}&=\mathcal{H}_{B_0}\otimes\mathcal{H}_{S_0}=\mathcal{H}_{B'_0}\otimes\mathcal{H}_{S'_0}.
\end{align}
Note that the input and output dimensions of the channels $\E_{AB}$ and $\E'_{A'B'}$ in consideration can in general be different. Therefore, the subsystems $E_0$, $R_0$ and $R'_0$ ($F_0$, $S_0$, and $S'_0$) shall be chosen accordingly such that the dimensions match. In particular, we set the dimension $d_{E_0}$ ($d_{F_0}$) to be a common multiple of the input dimensions $d_{A_0}$ and $d_{A'_0}$ ($d_{B_0}$ and $d_{B'_0}$). Namely, there exist natural numbers $n$, $n'$, $m$ and $m'$ so that 
\begin{align}
    d_{E_0}&=nd_{A_0}=n'd_{A'_0}\quad\text{and}\\
    d_{F_0}&=md_{B_0}=m'd_{B'_0}\,.
\end{align} 
Thus, the input states $\rho^x\in\mathcal{M}_{d_{E_0}}$ and $\sigma^y\in\mathcal{M}_{d_{F_0}}$ are regarded in different bipartitions when they are used as input states for the different channels $\E$ and $\E'$. For example, suppose we consider a channel $\E_{AB}$ with $d_{A_0}=4$ and a channel $\E'_{A'B'}$ with $d_{A'_0}=6$. Then we can consider an input state $\rho^x\in\mathcal{M}_{12}$, which can be seen either as a bipartite state $\rho^x_{A_0R_0}$ with $d_{A_0}=4$ and $d_{R_0}=3$, or as a bipartite state $\rho^x_{A'_0R'_0}$ with $d_{A'_0}=6$ and $d_{R'_0}=2$.

In a similar manner, consider the convex combinations of local measurements in the output systems $E_1$ and $F_1$,
$Z^{a,b}_{E_1F_1}=\sum_ip_i {M^{a,i}}_{E_1}\otimes {N^{b,i}}_{F_1}$. Following the same arguments in the proof of Proposition 1 of \cite{BuscemiNonlocHV_2012}, notice that the probability distributions available for a particular game with different measurements form a convex set. Therefore, $\gamma(\E)\geq\gamma(\E')$ for all $\gamma$ is equivalent to assuming for any measurement $Z$ of the above form  there exists a measurement $\overline{Z}$ in the convex hull of local measurements $\overline{M}_{E_1}\otimes\overline{N}_{F_1}$ such that
\begin{equation}\label{eq:appDistrGamesBuscemi1}
p(ab|xy)
=
\tr\Big [ \overline{Z}^{a,b}_{E_1F_1}\cdot\text{id}_{RS}\otimes\E_{AB}(\rho^x_{E_0}\otimes\sigma^y_{F_0}) \Big ] 
=
\tr\Big [ Z^{a,b}_{E_1F_1}\cdot\text{id}_{R'S'}\otimes\E'_{A'B'}(\rho^x_{E_0}\otimes\sigma^y_{F_0}) \Big ] \,.
\end{equation}
Define now the input states as
\begin{subequations}
\begin{align}
\rho^x_{E_0}&:=\tr_{E_0^c}[\Theta^x_{E_0^c}\otimes\id_{E_0}\cdot\phi^+_{E_0^cE_0}]\quad\text{and}\\
\sigma^y_{F_0}&:=\tr_{F_0^c}[\Upsilon^y_{F_0^c}\otimes\id_{F_0}\cdot\phi^+_{F_0^cF_0}]\, ,
\end{align}
\end{subequations}
with extended Hilbert spaces $\mathcal{H}_{E_0^c}\approx\mathcal{H}_{E_0}$ and $\mathcal{H}_{F_0^c}\approx\mathcal{H}_{F_0}$, where $\{\Theta^x\}_x$ and $\{\Upsilon^y\}_y$ are informationally complete, positive operator-valued measurements (IC-POVMs) and $\phi^+_{E_0^cE_0}$ is the density matrix of the Bell state $\ket{\phi^+}_{E_0^cE_0}=1/\sqrt{d_{E_0}}\sum_{i=0}^{d_{E_0}-1}\ket{i}_{E_0^c}\ket{i}_{E_0}$ (and similarly for $\phi^+_{F_0^cF_0}$). For convenience we will denote $\E_{EF}:=\text{id}_{RS}\otimes\E_{AB}$ and $\E'_{EF}:=\text{id}_{R'S'}\otimes\E'_{A'B'}$. Thus, we can write Eq.~\eqref{eq:appDistrGamesBuscemi1} as
\begin{align}
    p(ab|xy)&=\tr\Big [ \Theta^x_{E_0^c}\otimes\overline{Z}_{E_1F_1}\otimes\Upsilon^y_{F_0^c}\cdot\E_{EF}\otimes\text{id}_{E_0^cF_0^c}(\phi^+_{E_0E_0^c}\otimes\phi^+_{F_0F_0^c}) \Big ]\nonumber\\
    &=\tr\Big [ \Theta^x_{E_0^c}\otimes Z_{E_1F_1}\otimes\Upsilon^y_{F_0^c}\cdot\E'_{EF}\otimes\text{id}_{E_0^cF_0^c}(\phi^+_{E_0E_0^c}\otimes\phi^+_{F_0F_0^c}) \Big ]\,.
\end{align}
Now divide the full trace above as $\tr_{E_0^cF_0^c}[\tr_{E_1F_1}(\cdots)]$. By assumption, $\{\Theta^x\}_x$ and $\{\Upsilon^y\}_y$ are IC-POVMs and therefore form a basis for the linear operators acting on the Hilbert spaces $\mathcal{H}_{E_0^c}$ and $\mathcal{H}_{F_0^c}$. As a result, for the partial trace over $E_1$ and $F_1$, we have the following 
\begin{align}
    &\tr_{E_1F_1}\Big [ \id_{E_0^c}\otimes\overline{Z}^{ab}_{E_1F_1}\otimes\id_{F_0^c}\cdot\E_{EF}\otimes\text{id}_{E_0^cF_0^c}(\phi^+_{E_0E_0^c}\otimes\phi^+_{F_0F_0^c}) \Big ]=\nonumber\\
    &\tr_{E_1F_1}\Big [ \id_{E_0^c}\otimes Z^{ab}_{E_1F_1}\otimes\id_{F_0^c}\cdot\E'_{EF}\otimes\text{id}_{E_0^cF_0^c}(\phi^+_{E_0E_0^c}\otimes\phi^+_{F_0F_0^c}) \Big ]\,.
\end{align}
Let us now recover the notation $E_1=A_1R_1=A'_1R'_1$ and $F_1=B_1S_1=B'_1S'_1$ while noticing that $R_1=R_0$, $R'_1=R'_0$ and the same for $S$. Thus, we can write
\begin{align}
    &\tr_{\stackrel{\scriptstyle A_1R_0}{B_1S_0}}\Big [ \id_{A_0^cR_0^c}\otimes\overline{Z}^{ab}_{\stackrel{\scriptstyle A_1R_0}{B_1S_0}}\otimes\id_{B_0^cS_0^c}\cdot\underbrace{\E_{AB}\otimes\text{id}_{RS}\otimes\text{id}_{\stackrel{\scriptstyle A_0^cR_0^c}{B_0^cS_0^c}}\left(\phi^+_{A_0A_0^c}\otimes\phi^+_{B_0B_0^c}\otimes\phi^+_{R_0R_0^c}\otimes\phi^+_{S_0S_0^c}\right)}_{J_{A_1B_1A_0^cB_0^c}\otimes\phi^+_{R_0R_0^c}\otimes\phi^+_{S_0S_0^c}} \Big ]=\label{eq:appLHSteleport}\\
    &\tr_{\stackrel{\scriptstyle A'_1R'_0}{B'_1S'_0}}\Big [ \id_{{A'}_0^c {R'}_0^{c}}\otimes Z^{ab}_{\stackrel{\scriptstyle A'_1R'_0}{B'_1S'_0}}\otimes\id_{{B'}_0^c{S'}_0^c}\cdot \underbrace{\E'_{A'B'}\otimes\text{id}_{R'S'}\otimes\text{id}_{\stackrel{\scriptstyle {A'}_0^c{R'}_0^{c}}{{B'}_0^c{S'}_0^c}}\left(\phi^+_{A'_0{A'}_0^c}\otimes\phi^+_{B'_0{B'}_0^c}\otimes\phi^+_{R'_0{R'}_0^{c}}\otimes\phi^+_{S'_0{S'}_0^c}\right)}_{J'_{A'_1B'_1{A'}_0^c{B'}_0^c}\otimes\phi^+_{R'_0{R'}_0^{c}}\otimes\phi^+_{S'_0{S'}_0^c}} \Big ]\,.\label{eq:appRHSteleport}\nonumber
\end{align}
On each side of the equality, we identify the Choi states $J$ and $J'$ tensored by two pairs of Bell states $\phi^+$. This allows for a physical interpretation in terms of teleportation similar to~\cite{BuscemiNonlocHV_2012}, with the difference that now two parts out of the four-partite Choi state $J'$ are teleported. That is, here Alice holds the systems $A'_1$, $B'_1$, $R'_0$ and $S'_0$ and Bob holds the systems ${A'}_0^{c}$, ${B'}_0^{c}$, ${R'}_0^{c}$ and ${S'}_0^{c}$. Using two shared Bell states, $\phi^+_{R'_0{R'}_0^c}$ and $\phi^+_{S'_0{S'}_0^c}$, Alice teleports her parts $A'_1$ and $B'_1$ of the shared Choi state $J'$ to Bob's parts ${R'}_0^c$ and ${S'}_0^c$.

By assumption, for any $Z^{ab}$ there exists a $\overline{Z}^{ab}$ such that the equality  Eq.~\eqref{eq:appLHSteleport} holds. In particular, it holds when Alice applies a Bell measurement $\{Z^{ab}_{A'_1B'_1R'_0S'_0}\}=\{\phi^{ab}_{A'_1R'_0}\otimes\phi^{ab}_{B'_1S'_0}\}$ to send her side to Bob, where $\phi^{ab}$ is the density matrix of each generalized Bell state $\ket{\phi^{ab}}=\sum_{k=0}^{d-1}e^{bk\frac{2\pi i}{d}}\ket{k,k+a}/\sqrt{d}$~\cite{Teleport_Bennett1993}. 
In this case, Bob can apply unitaries $\{ U^a_{{R'}_0^c}\otimes V^b_{{S'}_0^c}\otimes\id_{{A'}^{c}_0{B'}^{c}_0}\}$  to recover $J'$ in his side. Strictly speaking,
\begin{align}
    J'_{{A'}_0^c{B'}_0^cA'_1B'_1}=\sum_{a,b}&(U^a_{{R'}_0^c}\otimes V^b_{{S'}_0^c}\otimes\id_{{A'}_0^c{B'}_0^c})\\ \nonumber
    &\tr_{\stackrel{\scriptstyle A_1R_0}{B_1S_0}}\Big [ \id_{A_0^cR_0^c}\otimes\overline{Z}^{ab}_{\stackrel{\scriptstyle A_1R_0}{B_1S_0}}\otimes\id_{B_0^cS_0^c}\cdot J_{A_1B_1A_0^cB_0^c}\otimes\phi^+_{R_0R_0^c}\otimes\phi^+_{S_0S_0^c}\Big ](U^a_{{R'}_0^c}\otimes V^b_{{S'}_0^c}\otimes\id_{{A'}_0^c{B'}_0^c})^\dagger,
\end{align}
where we inserted the left-hand side of Eq.~\eqref{eq:appLHSteleport} as the (unnormalized) state Bob obtains for each outcome $a$ and $b$ of Alice. 
Finally, expanding $\overline{Z}^{ab}=\sum_i\nu_i\overline{M}^{a,i}_{A_1R_0}\otimes\overline{N}^{b,i}_{B_1S_0}$ with $v_i\geq 0$ and $\sum_i\nu_i=1$, one gets the following maps
\begin{align}
    \Lambda_i
    (\tau_{A_1A^c_0})&:=\sum_a(U^a_{{R'}_0^c}\otimes\id_{{A'}^{c}_0})\cdot\tr_{A_1R_0}\big [ \id_{A_0^cR_0^c}\otimes\overline{M}^{a,i}_{A_1R_0}\cdot\phi^+_{R_0R_0^c}\otimes\tau_{A_1A^c_0} \big ]\cdot({U^a_{{R'}_0^c}}\otimes\id_{{A'}^{c}_0})^\dag\quad\text{and}\\
    \Phi_i    (\theta_{B_1B^c_0})&:=\sum_b(V^b_{{S'}_0^c}\otimes\id_{{B'}^{c}_0})\cdot\tr_{B_1S_0}\big [ \id_{B_0^cS_0^c}\otimes\overline{N}^{b,i}_{B_1S_0}\cdot\phi^+_{S_0S_0^c}\otimes\theta_{B_1B^c_0} \big ]\cdot({V^b_{{S'}_0^c}}\otimes\id_{{B'}^{c}_0})^\dag\,.
\end{align}
One can verify that $\tr_{{R'}^c_0} [\Lambda_i(\tau_{A_1 A_0^c})] = \id_{{A'}^c_0}$ if $\tr_{A_1}(\tau_{A_1A_0^c})=\id$, 
thus each $\Lambda_i$ is a quantum superchannel. 
Similarly, each $\Phi_i$ is a superchannel. 
It is now clear that the Choi matrix $J'$ can be obtained through local operations assisted by shared randomness from $J$,
\begin{equation}
    J'_{{A'}_0^c{B'}_0^cA'_1B'_1}=\sum_i\nu_i\Lambda_i\otimes\Phi_i(J_{A_1B_1A_0^cB_0^c}),
\end{equation}
which means that the channel $\E'$ can be obtained from the channel $\E$ by a convex combination of products of superchannels.
\end{proof}

\section{Proof of Proposition~\ref{prop:CAiffCorrSubs}}
\label{app:CAiffCorrSubs}
First we show the second item (ii). Let $p(a,b|x,y) = \tr (\sigma P^{a|x} \otimes Q^{b|y})$ determine a quantum probability distribution, where $\sigma$ is a quantum state on $RS$, $\{P^{a|x}\}_a$ projective measurements over $R$ and $\{Q^{b|y}\}_b$ over $S$. Going to a Stinespring's dilation \cite[Theorem 6.9]{holevo2019quantum}, there exist unitary maps $U_x: AR \rightarrow AR$ and $V_y : BS \rightarrow BS$ such that
\begin{equation}
P^{a|x} = \bra a U_x \ket 0 \quad \& \quad
Q^{b|y} = \bra b V_y \ket 0 \, .
\end{equation}
Take $U : CAR \rightarrow CAR$ and $V : DBS \rightarrow DBS$ unitaries on some control systems $C$ and $D$,
\begin{equation}
U = \sum_x \ket x \bra x \otimes U_x \quad \& \quad 
V = \sum_y \ket y \bra y \otimes V_y \, .
\end{equation}
These define a LOSE channel $\mathcal E: CD \rightarrow AB$ via
\begin{equation}
\mathcal E (\rho) = \tr_{CDRS} [U \otimes V (\rho \otimes \ketbra{00}{00}\otimes \sigma) U^\dagger \otimes V^\dagger] \,
\end{equation}
whose decoherent action gives back $p(a,b|x,y)$. Indeed,
\begin{align}
& \bra{ab} U \otimes V (\ketbra{xy}{xy} \otimes \ket{00}) = \ket{xy} \otimes P^{a|x} \otimes Q^{b|y} \, ,
\end{align}
thus $\tr [\ketbra{ab}{ab} \mathcal E(\ketbra{xy}{xy})]$ reads
\begin{align}
& \tr[\ketbra{xy}{xy} \otimes P^{a|x} \otimes Q^{b|y} \sigma] = p(a,b|x,y) \, .
\end{align}

Conversely, assume $p(a,b|x,y)$ can be obtained as the decoherent action of a LOSE channel, $\tr[(\ketbra{ab}{ab} \otimes \id) U\otimes V (\ketbra{xy}{xy} \otimes \sigma) U^\dagger \otimes V^\dagger]$, where $U:AR \rightarrow AR$ and $V: BS \rightarrow BS$ are unitaries and $\sigma$ a state on $RS$. With the isometries
\begin{equation}
U_x = U\ket x \quad \& \quad V_y = V\ket y \, 
\end{equation}
we can define quantum effects 
\begin{equation}
E^{a|x} = U_x^\dagger \ketbra{a}{a} U_x \quad \& \quad 
F^{b|y} = V_y^\dagger \ketbra{b}{b} V_y \, ,
\end{equation}
that reproduce the probabilities $p(a,b|x,y)$. Indeed,
\begin{equation}
\tr[U_x^\dagger \otimes V_y^\dagger \ketbra{xy}{xy} U_x \otimes V_y \sigma] = \tr[E^{a|x} \otimes F^{b|y} \sigma] \, .
\end{equation}

Item (i) follows from the above derivations by considering $\sigma$ to be a separable state. Item (iii) can be seen by realizing that the QNS conditions on a diagonal Choi matrix are equivalent to the NS conditions on the probability distribution given by the decoherent action of the channel.

\end{widetext}
\bibliographystyle{quantum_abbr}
\bibliography{Bibliography}

\end{document}